\documentclass[3p, preprint, times, fleqn]{elsarticle}	

\usepackage{mathtools}
\usepackage{amssymb}
\usepackage{upgreek}
\usepackage{derivative}
\usepackage{braket}

\usepackage{xcolor}
	\definecolor{Verde}{cmyk}{1,0.21,1,0.2}
	\definecolor{Blu}{cmyk}{1,0.6,0,0.2}
	
\usepackage{tikz}
	\usetikzlibrary{patterns}
	\tikzset{>=latex}
	
\usepackage{pgfplots}
	\usepgfplotslibrary{colormaps, fillbetween}
	\pgfplotsset{/pgf/number  format/1000 sep={\,}, compat=newest}
	\pgfkeys{/tikz/singletrim/.style={trim axis left, trim axis right}}
	\pgfplotsset{
		basicoptions/.style={
			tick label style={font=\footnotesize},
			label style={font=\small},
			legend style={font=\small},
			legend style={draw=none}, legend cell align=left,},
		singleplot/.style={
			basicoptions, width={8.5cm}, height={5.5cm}},
	}

\usepackage{hyperref}
\hypersetup{colorlinks}

\newcommand*{\vb}[1]{\boldsymbol{#1}} %vector bold
\NewDocumentCommand{\grad}{e{_^}}{%
  \mathop{}\!% \mathop for good spacing before \nabla
  \vb{\nabla}
  \IfValueT{#1}{_{\!#1}}% tuck in the subscript
  \IfValueT{#2}{^{#2}}% possible superscript
}
\newcommand*{\diver}[1]{\grad\cdot {#1}}

	\DeclarePairedDelimiterX\comm[2]{[}{]}{#1,\mathopen{} #2}

	\newcommand*{\R}{\mathbb{R}}

	\newcommand*{\Z}{\mathbb{Z}}	

	\newcommand*{\iu}{\mathrm{i}\mkern1mu} %imaginary unit
	\newcommand*{\e}{\mathrm{e}}
	\newcommand*{\pg}{\uppi} %pi

	\newcommand{\ie}{i.e.\ }
	\newcommand{\eg}{e.g.\ }

	\newcommand{\cc}{^{\ast}}

	\newcommand*{\A}{\vb{A}}
	\newcommand*{\Ap}{\skew{4}\tilde{\vb{A}}}
	\newcommand*{\App}{\tilde{A}}
	
	\newcommand*{\gradA}{\grad_{\A}}
	\newcommand*{\gradAp}{\grad_{\Ap}}
	\newcommand*{\n}{\vb{n}}
	\newcommand*{\tv}{\vb{t}}
	\newcommand*{\x}{\vb{x}}
	
	\newcommand*{\HA}{ H_{\A} }
	\newcommand*{\HAp}{ H_{\Ap} }
	
	\newcommand*{\B}{ \mathcal{B} }
	\newcommand*{\BA}{ \mathcal{B}_{\A} }
	
	\newcommand*{\BAp}{ \mathcal{B}_{\Ap} }
	\newcommand*{\nuA}{ \nu_{\!\A} }
	\newcommand*{\nuAp}{ \nu_{\!\Ap} }
	
	\newcommand*{\closure}[1]{\overline{#1}}

	\newcommand*{\dO}{\partial\Omega}
	\newcommand*{\Omegabar}{\closure{\Omega}}

	\newcommand*{\testO}{ \mathcal{D}(\Omegabar) }
	\newcommand*{\testdO}{ \mathcal{D}(\dO) }

\newtheorem{Proposition}{Proposition}
\newdefinition{Definition}{Definition}
\newproof{proof}{Proof}

\makeatletter
\def\ps@pprintTitle{%
 \let\@oddhead\@empty
 \let\@evenhead\@empty
 \def\@oddfoot{}%
 \let\@evenfoot\@oddfoot}
\makeatother

\begin{document}
\begin{frontmatter}
\title{Quantum magnetic billiards: boundary conditions and gauge transformations
%\tnoteref{t1,t2}
}
%\tnotetext[t1]{This document is the results of the research
%project funded by the National Science Foundation.}
%\tnotetext[t2]{The second title footnote which is a longer
%text matter to fill through the whole text width and
%overflow into another line in the footnotes area of the
%first page.}

\author[1,2]{Giuliano Angelone\corref{cor1}%\fnref{fn1}
	}\ead{giuliano.angelone@ba.infn.it}
\author[1,2]{Paolo Facchi%\fnref{fn2}
	}%\ead{}
\author[1,2]{Davide Lonigro%\fnref{fn1,fn3}
	}%\ead{}

\cortext[cor1]{Corresponding author.}
%\fntext[fn1]{This is the first author footnote.}
%\fntext[fn2]{Another author footnote, this is a very long
%footnote and it should be a really long footnote. But this
%footnote is not yet sufficiently long enough to make two
%lines of footnote text.}
%\fntext[fn3]{Yet another author footnote.}
\address[1]{Dipartimento di Fisica, Universit\`a di Bari, I-70126 Bari, Italy}
\address[2]{INFN, Sezione di Bari, I-70126 Bari, Italy}

\begin{abstract}
We study admissible boundary conditions for a charged quantum particle in a two-dimensional region subjected to an external magnetic field, \ie a quantum magnetic billiard. After reviewing some physically interesting classes of admissible boundary conditions (magnetic Robin and chiral boundary conditions), we turn our attention to the role of gauge transformations in a magnetic billiard: in particular, we introduce gauge covariant boundary conditions, and find a sufficient condition for gauge covariance which is satisfied by all the aforementioned examples.
\end{abstract}

\end{frontmatter}

%\linenumbers

\section{Introduction}

The study of quantum systems confined in a bounded domain requires a careful description of the physical properties of its boundary, and thus of the interaction between the system and the boundary, effectively encoded via a proper choice of boundary conditions. In recent years, quantum boundary conditions have increasingly attracted interest in different branches of quantum physics~\cite{AIM05,AIM15}, some examples being the analysis of quantum Hall systems~\cite{NT87, AANS98}, the study of geometric phases~\cite{FaGaMa16}, quantum control theory, topological insulators and QCD~\cite{ABPP12}, quantum gravity and topology change~\cite{ShWiXi12,PPBI15}, as well as the Casimir effect in quantum field theory~\cite{AsGaMu06,AMC13,ABPP16}, to name a few.

Quantum magnetic billiards are a paradigmatic example of such systems. Magnetic billiards consist of a charged particle which moves in a region of the plane and interacts solely with an external magnetic field and with the boundary of the region. Despite the apparently innocuous setup, even at the classical level, where the particle interaction with the billiard consists just in a reflection (an elastic scattering), magnetic billiards have a rich physical phenomenology as well as many interesting mathematical properties, which arise from the interaction between the arc-like trajectories of the particle and the reflection at the boundary~\cite{RobBer85,BerKun96}. Their quantum-mechanical counterpart is even more surprising, and throughout the years has been widely studied from different points of view and with various applications. 

From the physical perspective, quantum magnetic billiards represent a starting point to study magneto-transport properties and, in particular, the quantum Hall effect~\cite{MorSch17, Yoshi02}. Moreover, they constitute a non-trivial arena to understand the quantum-to-classical transition and the emergence of quantum chaos~\cite{Ser92, Naka02} as well as to investigate the subtleties related to the Aharonov-Bohm effect and, more generally, the interplay between the electromagnetic interaction and the geometry of the system~\cite{AhaBoh59, PeshTono89}. 
Quantum magnetic billiards are also extensively studied in the mathematical literature, which ranges e.g.\ from the rigorous analysis of magnetic Hamiltonians and of their spectral properties~\cite{Lein83, OliMon21}, to the description of gauge fields in  non-trivial topologies and geometries~\cite{PanRich11, WuYa75}.

Therefore, it is crucial to characterize the physically admissible boundary conditions for a quantum magnetic billiard. While, in general, a ``good" boundary condition must preserve the self-adjointness of the Hamiltonian (and thus, physically, ensure the overall conservation of probability), when dealing with magnetic systems an \emph{additional request} emerges: the \emph{gauge covariance} of the system must also be ensured, thus constraining the quantum boundary conditions of the magnetic billiard to be gauge covariant as well.

In this work, we apply the well-developed theory of self-adjoint extensions of Hermitian operators to a simple setup involving a non-relativistic spinless particle in a magnetic billiard. Following a bottom-to-top approach which focuses on the physical intuition behind the formalism, we will characterize a large family of physically admissible boundary conditions for a regular quantum magnetic billiard, also showing that such a family does include some remarkable examples of quantum boundary conditions, namely magnetic Robin and chiral ones. To do so, we shall develop a versatile framework capable of describing the relationship between the boundary conditions and the self-adjoint extension of the Hamiltonian, without dwelling too much on the mathematical subtelties related to the analysis of unbounded operators. Our approach also paves the way to some interesting applications, which will be explored in future works, such as a self-consistent model of the quantum Hall effect seen as a magnetic billiard~\cite{quanHall}.

The paper is organized as follows. In \autoref{sec-qmb} we describe our setup introducing the relevant Hamiltonian (i.e.\ the magnetic Laplacian). In \autoref{sec-bc} we describe some interesting families of boundary conditions which render the Hamiltonian self-adjoint. Finally, in \autoref{sec-gauge} we investigate the general relationship between boundary conditions and gauge transformations.

\section{Quantum magnetic billiard}\label{sec-qmb}
For our purposes, a magnetic billiard consists in a nonrelativistic spinless particle with mass $m$ and  electric charge $-e<0$, constrained in a region $\Omega\subseteq\R^2$ of the $xy$-plane and subjected to a magnetic field aligned with the $z$ axis and of modulus $B(x,y)$. The magnetic field is related to the \emph{vector potential} $\A=(A_x, A_y)$ by the relation
\begin{equation}\label{eq-Bcurl}
B=\frac{\partial A_y}{\partial x}-\frac{\partial A_x}{\partial y}\,.
\end{equation}
In order to avoid technical difficulties, in the following we will focus on the case of a \emph{regular magnetic billiard}, \ie we will assume that:
\begin{enumerate}
\item $\Omega$ is an open, connected and bounded subset of $\R^2$ with a smooth boundary $\dO$;
\item $\A(x,y)$ is a smooth vector field on the closure of $\Omega$;
\end{enumerate}
However, we will not assume $\Omega$ to be simply connected, allowing i.e. the possibility of magnetic billiards with one or more holes, an interesting setup which is also relevant for the description of the Aharonov-Bohm effect in a bounded domain.

\subsection{The magnetic Laplacian}
In a quantum mechanical setting, the system must be described by a properly chosen self-adjoint operator (the Hamiltonian) acting on the Hilbert space $L^2(\Omega)$ of complex square-integrable functions on the set $\Omega\subseteq\R^2$,  
endowed with the  scalar product 
\begin{equation}
	\braket{\psi|\phi}= \int_{\Omega} \psi(\x)^*\phi(\x)\,\mathrm{d}\x\,,
\end{equation}
where $\x=(x,y)$, and its associated norm $\|\psi\|^2= \braket{\psi|\psi}$.
Formally, the quantum Hamiltonian of a magnetic billiard is given by the differential operator
\begin{equation}\label{eq-HA}
\HA\equiv \frac{1}{2m}(-\iu\hbar\grad+e\A)^2=-\frac{\hbar^2}{2m}\gradA^2\,,
\end{equation}
where $\gradA^2$ is known as the magnetic Laplacian, $\hbar$
 is the reduced Planck constant, and 
\begin{equation}\label{eq-gradA}
\gradA\equiv \grad+\iu\frac{e}{\hbar}\A = \left(\frac{\partial}{\partial x} + \iu\frac{e}{\hbar}A_x ,\; \frac{\partial}{\partial y} + \iu\frac{e}{\hbar}A_y \right)
\end{equation}
is the so-called covariant derivative.
More explicitly, the operator~\eqref{eq-HA} can be written  as
\begin{equation}\label{eq-HAxy}
\HA=-\frac{\hbar^2}{2m}\left(\frac{\partial^2}{\partial x^2} +\frac{\partial^2}{\partial y^2}\right)
-\iu \frac{\hbar e}{ m}\left( A_x(\x) \frac{\partial}{\partial x}+ A_y(\x) \frac{\partial}{\partial x}\right)
-\iu \frac{\hbar e}{ 2 m}\left(\frac{\partial A_x(\x)}{\partial x} +\frac{\partial A_y(\x)}{\partial y} \right)
+\frac{e^2}{2m}\A^2(\x) \,.
\end{equation}

\paragraph{Gauge freedom} It is well-known that Eq.~\eqref{eq-Bcurl} does not completely fix the vector potential in terms of the magnetic field, leaving a so-called gauge freedom: since distinct vector potentials corresponding to the same magnetic field are associated in general with distinct Hamiltonians, this introduces an ambiguity in the description of the system. However, as we will discuss in \autoref{sec-gauge}, when there are no topological obstructions distinct Hamiltonians associated with the same magnetic field turn out to be unitarily equivalent. This is always the case when the particle is free to move in the whole plane or, more generally, when $\Omega$ is a simply connected region.

\subsection{Domain specification}\label{sec-sa}
Equation~\eqref{eq-HA} is not enough to define a Hamiltonian on $L^2(\Omega)$, since the formal expression $\HA\psi$ is not  a square-integrable function for every $\psi\in L^2(\Omega)$. As such, a \emph{domain specification} is also needed: we must choose a dense subspace $\mathfrak{D}\subset L^2(\Omega)$ of wavefunctions such that, for every $\psi\in\mathfrak{D}$, $\HA\psi$ is square-integrable. Necessarily, for each such choice, $\HA$ will be rendered as an \emph{unbounded} operator~\cite{Teschl}.
 We stress that, far from being a mathematical subtlety, the domain specification is, in general, crucial: many properties of a given unbounded operator, such as its spectrum, strongly depend on the domain $\mathfrak{D}$ where it acts: different domains correspond to different physical realizations of $\HA$. In this sense, the domain specification is physical.

Furthermore, in order to describe a \emph{physical observable} (and thus to generate a unitary evolution), the Hamiltonian must be a \emph{self-adjoint} operator. Indeed, we recall that self-adjointness is a necessary and sufficient request for a Hermitian operator to have a real spectrum. For this purpose, it is convenient to introduce the following
boundary form associated with $\HA$:
\begin{equation}
	\Lambda_{\A}(\psi,\phi)\equiv\braket{\psi|\HA\phi}-\braket{\HA\psi|\phi}
=\int_{\Omega} \left[\psi\cc \HA\phi-(\HA\psi)\cc\phi\right] \,\mathrm{d}\x\,.\label{eq-LambdaAdef}
\end{equation}
With this notation, the operator $\HA$ with domain $\mathfrak{D}$ will be
\begin{itemize}
	\item \emph{Hermitian}  if $\Lambda_{\A}(\psi,\phi)=0$ for all $\psi,\phi\in\mathfrak{D}$;
	\item \emph{self-adjoint} if, for any fixed $\psi\in\mathfrak{D}$, $\Lambda_{\A}(\psi,\phi)=0$ if and only if $\phi\in\mathfrak{D}$,
\end{itemize}
the second condition obviously implying the first one. For practical purposes, however, it is enough requiring $\HA$ to be \emph{essentially self-adjoint}, i.e. to have a unique self-adjoint extension:
as a matter of fact, this allows one to choose a space of suitably well-behaved functions as domain, without losing any physical information. 

For the kind of systems we are considering in this paper, let us distinguish two paradigmatic examples.
\begin{itemize}
\item When the particle is free to move in the whole $xy$-plane, \ie when $\Omega=\R^2$, a suitable choice for the domain $\mathfrak{D}$ which guarantees the essential self-adjointness of $\HA$ is the space of test functions, \ie smooth and compactly supported functions, denoted by $\mathcal{D}(\R^2)$. This is a well-known result which holds even for non-smooth (singular) vector potentials: see e.g.\ Theorem~2 of~\cite{Le81}.
\item When the particle is constrained in a regular billiard $\Omega\subset\R^2$, instead, though being Hermitian on the domain $\mathfrak{D}=\mathcal{D}(\Omega)$ (the space of test functions supported in $\Omega$) the magnetic Laplacian $\HA$  is not essentially self-adjoint, but rather admits infinitely many self-adjoint extensions~\cite{AIM15}.
\end{itemize}
As it turns out, the existence of multiple self-adjoint extensions in the latter case is strictly linked to the behavior of wavefunctions $\psi\in L^2(\Omega)$ on the boundary $\partial\Omega$. Functions in $\mathcal{D}(\Omega)$, being compactly supported, do not ``see'' the boundary of the billiard. In other words, no information about the physics at the boundary is encoded in them. If we want $\HA$ to describe a unique physical implementation of the billiard, we must specify a \emph{boundary condition (BC)}.

Following a physical insight, in the next section we elucidate this point by clarifying which BCs are actually physically relevant.

\section{Boundary conditions for a regular magnetic billiard}\label{sec-bc}
Consider the larger space $\mathcal{D}(\Omegabar)$ of functions which are smooth up to the boundary, $\Omegabar=\Omega\cup \partial\Omega$ being the closure of $\Omega$. The magnetic Laplacian $\HA$ with domain $\mathcal{D}(\Omegabar)$ is, in general, a non-Hermitian operator whose boundary form, introduced in the previous section, can be fully expressed in terms of boundary quantities through the well-known Green's formula.

\begin{Proposition}\label{prop-lackofsymmetry}
	For each $\psi,\phi \in \mathcal{D}(\Omegabar)$, the boundary form~\eqref{eq-LambdaAdef} can be expressed as
	\begin{equation}\label{eq-LambdaA}
	\Lambda_{\A}(\psi,\phi)=\int_{\dO} (\n\cdot\vb{j}_{\psi,\phi}^{\A})(s) \,\mathrm{d}s\,,
	\end{equation}
	where $s$ denotes the induced curvilinear coordinate on $\dO$,  $\n=\n(s)$ is the unit normal vector and
	\begin{equation}\label{eq-probabilitycurrent}
	\vb{j}_{\psi,\phi}^{\A}\equiv-\frac{\hbar^2}{2m}[\psi\cc\gradA\phi -(\gradA\psi)\cc\phi]\,.
	\end{equation}
\end{Proposition}
\begin{proof}
	By inserting the explicit expression~\eqref{eq-HA} of $\HA$ into~\eqref{eq-LambdaAdef} and observing that the terms proportional to $\A^2$ cancel out, we obtain
	\begin{align}
	\Lambda_{\A}(\psi,\phi)&=\frac{1}{2m}\int_{\Omega} \left\{\psi\cc(-\iu\hbar\grad+e\A)^2\phi-[(-\iu\hbar\grad+e\A)^2\psi]\cc\phi\right\}\,\mathrm{d}\x
	\notag\\ 
	&=\frac{-\hbar^2}{2m}\int_{\Omega}\left\{\psi\cc (\diver{\gradA}+\iu \tfrac{e}{\hbar}\A\cdot\grad)\phi
	-[(\diver{\gradA} + \iu\frac{e}{\hbar}\A\cdot\grad)\psi]\cc\phi\right\}\,\mathrm{d}\x\notag \\	
	&=\frac{-\hbar^2}{2m}\int_{\Omega} 
	\diver{[\psi\cc \gradA\phi - (\gradA\psi)\cc\phi]}\,\mathrm{d}\x=\int_{\Omega} 
	\diver{\vb{j}_{\psi,\phi}^{\A}(\x)}\,\mathrm{d}\x\,. \label{eq-LambdaAbulk}
	\end{align}
	By applying the Gauss-Green theorem, the bulk integral~\eqref{eq-LambdaAbulk} is transformed into the line integral~\eqref{eq-LambdaA}, thus proving the claim.
\qed\end{proof}

The boundary form~\eqref{eq-LambdaA}, which depends only on the \emph{edge} behaviour of the wavefunctions and of the vector potential, has a clear physical interpretation: the quadratic form $\Lambda_{\A}(\psi,\psi)$ represents indeed the net flux through the boundary of the probability current density  $\cramped{\vb{j}_{\psi,\psi}^{\A}}$
associated with the state $\psi$. As a result, requiring $\HA$ to be Hermitian, \ie $\Lambda_{\A}(\psi,\psi)=0$ for all $\psi$ in its domain, corresponds to a probability conservation law~\cite{AMC12,AGAMC15}
 . 

Moreover, as a consequence of Proposition~\ref{prop-lackofsymmetry}, the domain $\mathfrak{D}$ for which the magnetic Laplacian $\HA$ is Hermitian can be determined by just looking at the behavior of the wavefunctions $\psi\in\mathfrak{D}$ on the boundary $\dO$. To  formalize this concept, let us give the following definitions.
\begin{figure}[tb]
	\centering
	\includegraphics{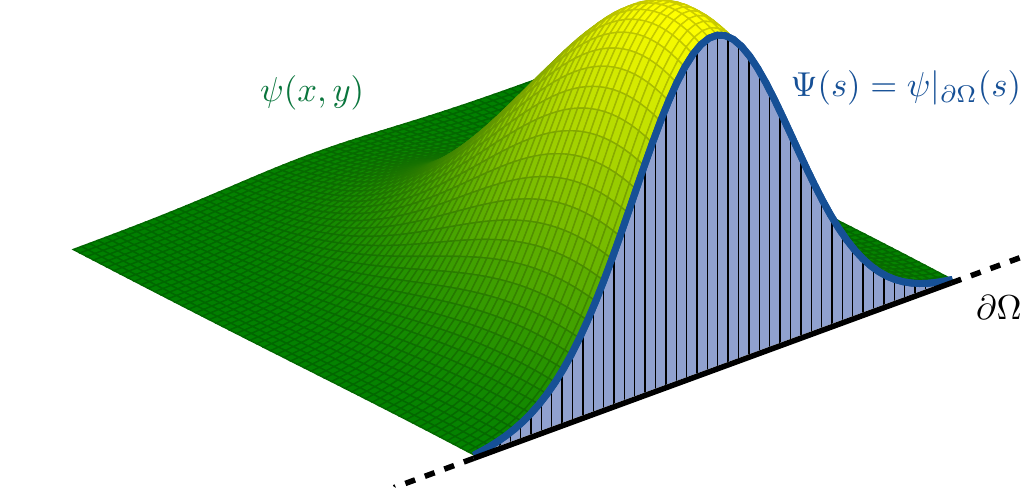}
	\caption{Example of a bulk function $\psi(x,y)\in\testO$ (in green) and of its restriction $\Psi(s)\in\testdO$ (in blue).}
	\label{fig-restriction}
\end{figure}

\begin{Definition}
	The \emph{restriction operator} $\gamma$ and the \emph{covariant normal derivative} $\nuA$ are defined respectively by
	\begin{align}
	\gamma\colon\testO\to\testdO\,,& \qquad\gamma\colon \psi \mapsto \Psi\equiv \psi|_{\dO}\,;\label{eq:gamma}\\
	\nuA\colon\testO\to\testdO\,,& \qquad\nuA\colon  \psi\mapsto \dot\Psi\equiv\n\cdot(\gradA{\psi})|_{\dO}\,.\label{eq:nuA}
	\end{align} 
\end{Definition}

To highlight the difference between $\testO$ and $\testdO$, sketched in \autoref{fig-restriction}, we will reserve the upper-case Greek letters to denote wavefunctions defined exclusively on the boundary. With the above definitions, the boundary form $\Lambda_{\A}$ can be equivalently written as
\begin{equation}
\Lambda_{\A}(\psi,\phi)=-\frac{\hbar^2}{2m}\int_{\dO}\left[\gamma\psi^*\,\nuA\phi-(\nuA\psi)^*\,\gamma\phi\right](s)\,\mathrm{d}s
=-\frac{\hbar^2}{2m}\int_{\dO}\left[\Psi^*\dot\Phi-\dot{\Psi}^*\Phi\right](s)\,\mathrm{d}s\,.
\end{equation}
This expression strongly suggests to implement a BC  as a linear relation  between the values of $\Psi=\gamma\psi$ and $\dot\Psi=\nuA\psi$, expressed through a suitable bulk-to-boundary operator.\footnote{Here we shall only consider BCs defined via \emph{bounded} bulk-to-boundary operators. Such BCs will be enough for our purposes. In general, as briefly discussed in~\ref{app:generalboundary}, unbounded operators may be also taken into account.}
\begin{Definition}
	A \emph{bulk-to-boundary operator} is an operator $\BA\colon\testO\to\testdO$ which can be expressed in the form
	\begin{equation}\label{eq-BBop}
	\BA=T_{1,\A} \gamma - T_{2,\A} \nuA\,,
	\end{equation}
	where $T_{1,\A}$, $T_{2,\A}$ are two boundary operators, that is
	\begin{equation}\label{eq-Top}
	T_{i,\A}\colon \testdO\to \testdO\qquad  (i=1,2)\,.
	\end{equation}
\end{Definition}

The bulk-to-boundary operator $\BA$ acts thus between two different spaces, mapping a function on the bulk $\psi\in \testO$ to a  function on the boundary $\BA\psi\in \testdO$. 
\begin{Definition}\label{def-bc}
	Given a vector potential $\A$, a \emph{magnetic boundary condition} is a constraint $\BA\psi=0$, \ie a condition
	\begin{equation}\label{eq-BAcondition}
	\psi \in \ker \BA=\left\{\, \psi\in\testO : T_{1,\A}\Psi = T_{2,\A} \dot\Psi \,\right\}\,,
	\end{equation}
	where $\BA\colon \testO \to \testdO $ is a bulk-to-boundary operator.
\end{Definition}

It should be noted that a magnetic BC is associated to a bulk-to-boundary operator $\BA$ only up to an invertible operator $C$, since $\ker\BA=\ker(C\BA)$. In what follows, with a slight abuse of notation we will eventually denote a bulk-to-boundary operator itself as the corresponding BC. Notice also that the kernel $\ker\BA$ is never a trivial space, as it always contains the whole space $\mathcal{D}(\Omega)$ of test functions on $\Omega$. In our formalism, for every BC $\BA$ we have
$	\mathcal{D}(\Omega)\subset\ker\BA\subset\testO$,
the magnetic Laplacian being Hermitian but not essentially self-adjoint on $\mathcal{D}(\Omega)$, and non-Hermitian on $\testO$. The Hamiltonian obtained by imposing a magnetic BC on $\testO$ “lays in the middle”.

\paragraph{Physically admissible boundary conditions}
An admissible magnetic BC, \ie one associated with a physical observable, must necessarily render $\HA$ essentially self-adjoint as an operator with domain $\ker\BA$. Such BCs can be parametrized by unitary operators acting on a proper boundary space~\cite{AIM05,ILP15,FGL18}. This point is briefly discussed in~\ref{app:generalboundary}.
Notice, however, that this is \emph{not} enough: the dependence of the boundary operators $T_{i,\A}$ on the vector potential $\A$ is not completely arbitrary, since a good magnetic BC should also be \emph{gauge covariant}. This crucial point will be discussed in~\autoref{sec-gaugebc}.

We conclude this section by introducing some interesting examples of magnetic BCs. To keep our discussion simple, we will only verify the Hermiticity of $\HA$ subjected to our selected BCs, relying to the more general results presented in~\ref{app:generalboundary} to prove essential self-adjointness. 

\subsection{Magnetic Robin boundary conditions}
\begin{Definition}
Let $\alpha\colon \dO\to\R$ be a smooth, real-valued function on the boundary $\partial\Omega$ of the billiard $\Omega$. The \emph{magnetic Robin boundary conditions} $\BA^{\textup{R}}(\alpha)$ are defined via
\begin{equation}
T_{1,\A}=\alpha I \qquad\text{and}\qquad T_{2,\A}=I\,,
\end{equation}
that is
\begin{equation}\label{eq-Robin}
\begin{aligned}
\psi \in \ker\BA^{\textup{R}}(\alpha)\iff\n\cdot(\gradA{\psi})|_{\dO}(s)=\alpha(s)\,\psi|_{\dO}(s)\,.
\end{aligned}
\end{equation}
\end{Definition}
Physically, magnetic Robin BCs encode a local proportionality between the boundary values of the functions and their magnetic covariant normal derivatives. In particular, at a boundary point $s$ when $\n\cdot\A=0$ so that $\n\cdot\gradA=\n\cdot\grad$, Robin BCs are repulsive or attractive in the cases $\alpha(s)<0$ and $\alpha(s)>0$, respectively: see the example in~\autoref{fig:Robin}.

\begin{figure}[tb]
	\centering
	\includegraphics{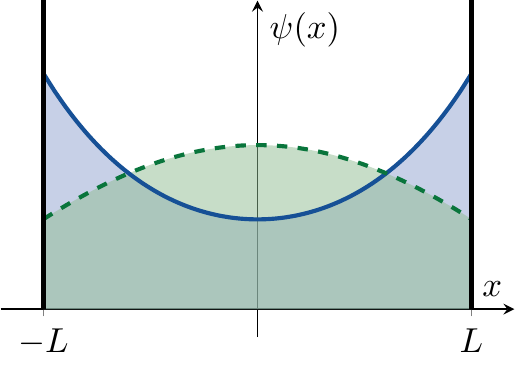}
	\caption{Attractive versus repulsive Robin boundary conditions: Ground energy eigenfunctions of the free Laplacian on $\Omega=\left(-L,L\right)$, subjected to Robin BCs $\n\cdot \grad\psi|_{\dO}=\alpha \psi|_{\dO}$ with $\alpha L=2$ (solid blue line) and $\alpha L=-2$ (dashed green line).}
	\label{fig:Robin}
\end{figure}

\begin{Proposition}\label{prop-Robin}
	The magnetic Hamiltonian $\HA$ with Robin BCs $\BA^{\textup{R}}(\alpha)$ is  Hermitian for every choice of gauge. 
\end{Proposition}

Interestingly, this property holds independently of the geometry of $\Omega$ and hence independently of the properties of its boundary $\dO$.
\begin{proof}
	Let us suppose that $\psi, \phi\in\ker\BA^{\textup{R}}(\alpha)$, so that both satisfy~\eqref{eq-Robin}. Then
	\begin{equation}
	\Lambda_{\A}(\psi,\phi) =-\frac{\hbar^2}{2m}\int_{\dO}\left[\psi\cc\n\cdot\gradA\phi -\n\cdot(\gradA\psi)\cc\phi\right]\,\mathrm{d}s
	=-\frac{\hbar^2}{2m}\int_{\dO}\left[\alpha\psi\cc\phi -\alpha\psi\cc\phi\right]\,\mathrm{d}s=0\,,
	\end{equation}
	whence the claim. \qed
\end{proof}

Robin BCs admit two familiar limit cases given by the limits $\alpha(s)\to\infty$ and $\alpha(s)=0$, corresponding respectively to  \emph{Dirichlet boundary conditions} $\B^{\textup{D}}$, obtained by setting $T_{1,\A}=I$ and $T_{2,\A}=0$ so that
\begin{equation}\label{eq-Dirichlet}
\psi \in \ker\B^{\textup{D}} \iff \Psi=0 \,,
\end{equation}
and to  \emph{magnetic Neumann conditions} $\BA^{\textup{N}}$, obtained by setting $T_{1,\A}=0$ and $T_{2,\A}=I$:
\begin{equation}\label{eq-Neumann}
\psi\in\ker\BA^{\textup{N}}\iff \dot\Psi=0\,.
\end{equation}
Notice that $\B^{\textup{D}}$ is the only magnetic Robin BC not involving the magnetic potential $\A$. For a general function $\alpha(s)$, Robin BCs clearly represent a (pointwise) interpolation between Dirichlet and Neumann BCs.

\subsection{Magnetic chiral boundary conditions}
\begin{Definition} Let $\alpha\colon \dO\to\R$ be a smooth, real-valued function on the boundary, and let $\beta\in\mathbb{R}$ be a real constant. Denote the unit tangent vector on $\partial\Omega$ by $\tv=\tv(s)$. The \emph{magnetic chiral boundary conditions} $\BA^{\textup{C}}(\alpha,\beta)$ are defined via
\begin{equation}
	T_{1,\A}=\alpha I+\iu\beta\left(\frac{\mathrm{d}}{\mathrm{d}s}+\iu\frac{e}{\hbar} \tv\cdot\A\right)\qquad\text{and}\qquad T_{2,\A}=I\,,
\end{equation}
\ie
\begin{equation}\label{eq-chiral}
	\psi\in\ker\BA^{\textup{C}}(\alpha, \beta)\iff \n\cdot(\gradA{\psi})|_{\dO}(s)=\alpha(s)\psi|_{\dO}(s)+\iu\beta\tv\cdot(\gradA{\psi})|_{\dO}(s)\,.
\end{equation} 
\end{Definition}

These BCs generalize Robin BCs, which are obtained by setting $\beta=0$. For $\beta\neq0$, they relate the value of a function and of its normal derivative to its tangential derivative, which physically corresponds to the particle velocity along the boundary. Magnetic chiral BCs on a disk have been studied in~\cite{JJV95}. In particolar, for some values of $\beta$, the spectrum of $\HA$ with such BCs is shown to be doubly unbounded, \ie the system is unstable.
\begin{Proposition}\label{prop-chiral}
	The magnetic Hamiltonian $\HA$ with chiral BCs $\BA^{\textup{C}}(\alpha, \beta)$  is  Hermitian for every choice of gauge. 
\end{Proposition}

As for the Robin case, this property is remarkably true independently of the geometry of the system.
\begin{proof}
	Let us suppose that $\psi, \phi\in\ker\BA^{\textup{C}}(\alpha, \beta)$ so that both satisfy~\eqref{eq-chiral}. Then
	\begin{align}
	\Lambda_{\A}(\psi,\phi)&=-\frac{\hbar^2}{2m} \int_{\dO}\left[\psi\cc\n\cdot\gradA\phi -\n\cdot(\gradA\psi)\cc\phi\right]\,\mathrm{d}s\notag\\
	&=-\frac{\hbar^2}{2m}\int_{\dO}\left\{\psi\cc[\alpha\phi+\iu\beta\tv\cdot\gradA\phi] - [\alpha\psi\cc-\iu\beta\tv\cdot(\gradA \psi)\cc]\phi\right\}\,\mathrm{d}s\notag\\
	&=-\frac{\hbar^2}{2m}\int_{\dO}\left\{\psi\cc\left[\iu \beta\frac{\partial\phi}{\partial s}-\beta\frac{e}{\hbar}\tv\cdot\A \phi\right] + \left[\iu\beta\frac{\partial\psi\cc}{\partial s}+\beta\frac{e}{\hbar}\tv\cdot\A \psi\cc\right]\phi\right\}\,\mathrm{d}s=-\frac{\hbar^2}{2m}\iu\beta\int_{\dO} \frac{\partial}{\partial s}(\psi\cc\phi)\,\mathrm{d}s\,.
\end{align}
Because of our regularity assumptions and the property $\partial(\dO)=\emptyset$, applying the Gauss-Green theorem we finally have
	\begin{equation}
	\Lambda_{\A}(\psi, \phi)=-\frac{\hbar^2}{2m}\iu\beta\int_{\partial(\dO)} \psi\cc\phi\,\mathrm{d}\xi=0\,,
	\end{equation}
	henceforth the claim. \qed
\end{proof}

\subsection{Non-local boundary conditions}
Robin and chiral BCs have a local nature, since they relate the value of the wave function at each point of the boundary to its derivatives at the same point. However, admissible BCs can also be non-local, and connect  the values of the wave function and of its derivatives at a boundary point $s$ to their values at  different boundary points $s'\neq s$. In this case the global geometry of $\Omega$ plays a crucial role, as non-local BCs as a matter of fact change the topology of the system: many interesting examples are considered e.g.\ in~\cite{AIM05,ShWiXi12}. A precise formulation of non-local BCs is however beyond the scope of this paper, and we refer the interested reader to the construction presented in~\cite{IbPP15}.
	
\section{Gauge transformations and magnetic boundary conditions}\label{sec-gauge}
We will  finally discuss the role of gauge transformations in the description of a quantum magnetic billiard. As remarked in~\autoref{sec-qmb}, different vector potentials linked by a gauge transformation, thus corresponding to the same magnetic field, yield different Hamiltonians which, in principle, may correspond to different physical dynamics. Understanding the role of magnetic BCs in this framework is the main contribution of our work. Before focusing on the BCs in \autoref{sec-gaugebc}, however, let us give a definition of a gauge transformation that can accommodate also the case of a non-simply-connected billiard.

\subsection{Gauge transformations}
\begin{Definition}
	Two vector potentials $\A=(A_x, A_y)$ and $\Ap=(\App_x,\App_y)$ are connected by a \emph{gauge transformation} on a set $\Omega$ if 	\begin{equation}\label{eq-Acurl}
		\frac{\partial A_x}{\partial x}-\frac{\partial A_y}{\partial y} =\frac{\partial \App_x}{\partial x}-\frac{\partial\App_y}{\partial y}
	\end{equation}
	for each $\x\in\Omega$, \emph{and}
	\begin{equation}\label{eq-Aflux}
		\frac{e}{\hbar}\oint_{\gamma}\bigl[\A(s)-\Ap(s)\bigr]\cdot\mathrm{d}\vb{s}\in 2\pg\,\Z 
	\end{equation}
	for each closed path $\gamma$ contained in $\Omega$. Moreover, an $\A$-dependent quantity which transforms unitarily (resp. trivially) under a gauge transformation is said to be \emph{gauge covariant} (resp. \emph{gauge invariant}).
\end{Definition}

From a physical point of view, the condition in Eq.~\eqref{eq-Acurl} says that the magnetic fields respectively associated with $\A$ and $\Ap$ have to coincide \emph{inside} the billiard. These magnetic fields, which in principle  may differ \emph{outside} the billiard,  are however further related by Eq.~\eqref{eq-Aflux}, a topological condition which expresses a flux quantization. Eq.~\eqref{eq-Aflux} says indeed that the magnetic fluxes (through the same surface) associated with $\A$ and $\Ap$, respectively, can differ at most by an integer multiple of the magnetic flux quantum $2\pg \hbar/e=h/e$. In other words, if we consider a billiard with a hole and we look exclusively at what happens inside the billiard, we are free to add an arbitrary number of magnetic flux quanta through the hole, as we cannot distinguish in any way the two systems (see Proposition~\ref{prop-unitaryequivalence}).

Notably, the condition supplied by Eq.~\eqref{eq-Aflux} is not necessary (i.e.\ it is trivial) in the particular case where $\Omega$ is simply connected, since in this case Eq.~\eqref{eq-Acurl} implies that
\begin{equation}
		\oint_{\gamma}\bigl[\A(s)-\Ap(s)\bigr]\cdot\mathrm{d}\vb{s}=0\,.
\end{equation}
This is a consequence of a standard topological result, namely the fact that each closed one-form defined in a simply connected manifold is also exact. 

From the mathematical point of view, the conditions~\eqref{eq-Acurl}--\eqref{eq-Aflux} further imply that the quantity $\A-\Ap$ can be unambiguously characterized in terms of a suitable ``gauge function''. This is expressed by the following proposition, whose proof is based e.g.\ on Proposition~2.1.3 of~\cite{FoHe10}.
\begin{Proposition}\label{prop-gauge}
	Let $\A$, $\Ap$ be two vector potentials connected by a gauge transformation on a set $\Omega$. Then there exists a \emph{gauge function} $\chi$, in general multivalued, such that for each $\x\in \Omega$:
	\begin{enumerate}
	\item 	 both the quantities $\grad\chi$ and $\e^{-\iu e\chi/\hbar}$ are well-defined (i.e.\ ordinary) functions;
	\item we have $\Ap=\A-\grad\chi$.
	\end{enumerate}
\end{Proposition}
\begin{proof}
Let us fix a reference point $\x_0\in \Omega$. The quantity
\begin{equation}
\chi_\gamma(\x)\equiv \int_{\gamma(\x_0, \x)} \bigl[\A(s)-\Ap(s)\bigr]\cdot\mathrm{d}\vb{s}\,,
\end{equation}
where $\gamma(\x_0, \x)$ denotes a path in $\Omega$ joining $\x_0$ to $\x$, is in general a multivalued function that depends on the path $\gamma$. However, since $\A$ and $\Ap$ are connected by a gauge transformation, the difference $\chi_{\gamma}(\x)-\chi_{\gamma'}(\x)$ is an integer multiple of the magnetic flux quantum: 
\begin{align}
\chi_{\gamma}(\x)-\chi_{\gamma'}(\x)&=
\int_{\gamma(\x_0, \x)} \bigl[\A(s)-\Ap(s)\bigr]\cdot\mathrm{d}\vb{s}-
\int_{\gamma'(\x_0, \x)} \bigl[\A(s)-\Ap(s)\bigr]\cdot\mathrm{d}\vb{s}=\oint_{\gamma''}\bigl[\A(s)-\Ap(s)\bigr]\cdot\mathrm{d}\vb{s}
=\frac{2n\pg\hbar}{e}\,,
\end{align}
for some $n\in\Z$, where the closed path $\gamma''$ is given by $\gamma(\x_0, \x)\cup \gamma'(\x, \x_0)$. Therefore, for each pair of paths $\gamma$ and $\gamma'$ in $\Omega$, we get that
\begin{equation}
\grad\chi_\gamma(\x)=\grad\chi_{\gamma'}(\x)\qquad\text{and}\qquad \e^{-\iu e\chi_\gamma(\x)/\hbar}=\e^{-\iu e\chi_{\gamma'}(\x)/\hbar}\,,
\end{equation}
that is, the quantities $\grad\chi_\gamma(\x)$ and $\e^{-\iu e\chi_\gamma(\x)/\hbar}$ do not actually depend on the path, and are thus ordinary functions. The relation  $\Ap(\x)=\A(\x)-\grad\chi(\x)$ follows then immediately for each $\x\in\Omega$. \qed
\end{proof}

With this property, we are now able to describe a gauge transformation at the level of the Hilbert space, and to consequently evaluate its effects on the observables and on the states of the physical system. In particular, we can now prove the gauge covariance of the magnetic Laplacian $\HA$.
\begin{Proposition}\label{prop-unitaryequivalence}
	Let $\HA$ and $\HAp$ be the operators
	\begin{equation}
		\HA=-\frac{\hbar^2}{2m} \gradA^2 \quad\text{and}\quad \HAp=-\frac{\hbar^2}{2m}\gradAp^2\,,
	\end{equation}
	defined respectively on the domains $\mathfrak{D}(\HA)$ and  $\mathfrak{D}(\HAp)$. If the vector potentials $\A$ and $\Ap$ are connected by a gauge transformation, with gauge function $\chi$, and if 
	\begin{equation}\label{eq-DHAp}
		\mathfrak{D}(\HAp)=U_\chi \mathfrak{D}(\HA)= \left\{U_\chi\psi : \psi \in \mathfrak{D}(\HA)\,\right\}\,,
	\end{equation}
	with $U_\chi\equiv \e^{-\iu e\chi/\hbar}$, then $\HA$ and $\HAp$ are gauge covariant, \ie unitarily equivalent:
	\begin{equation}
	\HAp=U_\chi\HA U_\chi^\dagger\,.
	\end{equation}
	In particular, if $\HA$ is essentially self-adjoint, then $\HAp$ is also essentially self-adjoint.
\end{Proposition}
\begin{proof}
	The multiplication operator $U_\chi\colon L^2(\Omega)\to L^2(\Omega)$ defined as $U_\chi(\x)=\e^{-\iu e\chi(\x)/\hbar}$ is well-defined, as a consequence of Proposition~\ref{prop-gauge}, and obviously unitary. The unitary equivalence of $\HA$ and $\HAp$ follows then immediately from
		\begin{equation}\gradAp^2=\gradAp\cdot\gradAp=(U_\chi\gradA U_\chi^\dagger)\cdot(U_\chi\gradA U_\chi^\dagger)=U_\chi \gradA\cdot\gradA U_\chi^\dagger=U_\chi \gradA^2 U_\chi^\dagger\,,
		\end{equation}	
		that is, from the gauge covariance of $\gradA$.
	\end{proof}
	
If the vector potentials are smooth, the condition~\eqref{eq-DHAp} is trivially satisfied when $\HA$ and $\HAp$ are both defined on $\testO$. In particular this is the relevant result when the particle is free to move in the whole plane $\Omega=\R^2$, as we discussed in \autoref{sec-sa}. 

Moreover, if $\Omega$ is not simply connected, $\HA$ and $\HAp$ may \emph{not} be unitarily equivalent, having hence different spectral properties, even if $\A$ and $\Ap$ describe the same magnetic field inside the billiard. This phenomenon is particularly relevant for the description (and the interpretation) of the Aharonov-Bohm effect, see \eg ~\cite{WuYa75, BawBur83, MoRa}.

\subsection{Covariance of magnetic boundary conditions}\label{sec-gaugebc}
For a regular magnetic billiard, the discussion is complicated by the fact that the typical domain of the Hamiltonian involves a BC, namely $\mathfrak{D}(\HA)=\ker\BA$: in general, even if $\HAp$ is defined on $\mathfrak{D}(\HAp)=\ker\BAp$ and $\A$ is connected by a gauge transformation to $\Ap$, $\HA$ and $\HAp$ may not be unitarily equivalent. As previously anticipated, in order to ensure the unitary equivalence of $\HA$ and $\HAp$, admissible magnetic BCs should be gauge covariant, in the sense that
	\begin{equation}
	\ker\BAp=U_\chi\ker\BA\,,
	\end{equation}
	so that~\eqref{eq-DHAp} is satisfied. In the next proposition we give a condition on the boundary operators $T_{i,\A}$ introduced in Eq.~\eqref{eq-Top} which makes the corresponding BC covariant.
	\begin{Proposition}
		Let $\BA$, $\BAp$ be two magnetic boundary conditions associated with two vector potentials $\A$ and $\Ap$ connected by a gauge transformation, with gauge function $\chi$. If
		\begin{equation}
		T_{i,\Ap}=u_{\chi} T_{i,\A}u_{\chi}^\dagger\qquad(i=1,2)\,,
		\end{equation}
		where $u_{\chi}\equiv \gamma U_{\chi}$, then the boundary conditions are gauge covariant: 
		\begin{equation}
		\ker \BAp=U_\chi \ker\BA\,.
		\end{equation}
	\end{Proposition}
	\begin{proof}
	An immediate calculation shows that the covariant normal derivative $\nuA$ is actually gauge covariant, in $\testO$: 
		\begin{equation}
		\nu_{\Ap}U_\chi\psi=u_\chi\nuA\psi\,.
		\label{eq:nuA-gauge}
		\end{equation}
		Using this result and the gauge covariance of the $T_{i,\A}$ we get
		\begin{align}
		\BAp U_\chi\psi&=[T_{1,\Ap}\gamma - T_{2,\Ap} \nuAp]U_{\chi}\psi
		=[T_{1,\Ap}\gamma U_{\chi} - T_{2,\Ap} \nuAp U_{\chi}]\psi\notag\\
		&=[T_{1,\Ap}u_{\chi}\gamma  - T_{2,\Ap} u_{\chi}\nuA]\psi
		=u_{\chi} [T_{1,\A} \gamma - T_{2,\A}\nuA]\psi
		=u_{\chi}\BA \psi\,,
		\end{align}
		thus implying the claim. \qed
	\end{proof}
	
	We conclude this section by proving the gauge invariance of the boundary form $\Lambda_{\A}(\psi,\phi)$ for an arbitrary regular magnetic billiard.
	\begin{Proposition}\label{prop-Lambdacovariance}
	Let $\A$ and $\Ap$ be two vector potentials connected by a gauge transformation, with gauge function $\chi$. Then we have
		\begin{equation}
		\forall\,\psi, \phi\in \testO:\qquad
		\Lambda_{\A}(\psi,\phi)=\Lambda_{\Ap}(U_{\chi}\psi,U_{\chi}\phi)\,.
		\end{equation}
	\end{Proposition}
	\begin{proof}
		It follows from a direct calculation:
		\begin{align}
		\vb{j}_{U_\chi\psi,U_\chi\phi}^{\Ap}&=-\frac{\hbar^2}{2m}\left\{
		 (U_\chi\psi)\cc \gradAp U_\chi\phi -U_\chi\phi(\gradAp U_\chi\psi)\cc	 \right\} =-\frac{\hbar^2}{2m}\left\{U_\chi^\dagger\psi\cc (U_\chi\gradA U_\chi^\dagger) U_\chi\phi - [(U_\chi\gradA U_\chi^\dagger)U_\chi \psi]\cc U_\chi\phi\right\}\notag\\
		&=-\frac{\hbar^2}{2m}\left\{\psi\cc \gradA \phi-(\gradA\psi)\cc \phi\right\}
		=\vb{j}_{\psi,\phi}^{\A}\,. 
		\end{align}\qed
	\end{proof}
	
	This proposition ensures  that if $\Lambda_{\A}(\psi,\phi)$ vanishes for each $\psi,\phi$ belonging to $\ker \BA$, so does $\Lambda_{\Ap}(U_\chi\psi,U_\chi\phi)$ for each $U_\chi\psi, U_\chi\phi$ belonging to $\ker \BAp$. As a matter of facts, this result explains on general grounds why the Hamiltonian $\HA$ with either magnetic Robin or chiral BCs is Hermitian \emph{for each  choice of the gauge}, as shown in Propositions~\ref{prop-Robin} and~\ref{prop-chiral}.

\section{Conclusions}
In this work we investigated the role of boundary conditions for the magnetic Laplacian in bounded regions of the plane, examining some physically interesting examples and remarking the importance of considering  gauge covariant boundary conditions. In particular, we obtained sufficient conditions which ensure the gauge covariance of the bulk-to-boundary operator. Similar results can be analogously obtained for more general setups, such as relativistic billiards described by (magnetic) Dirac operators. 

For what concerns the magnetic Laplacian, a natural application of our results is to  characterize the dependence of the spectral properties on the boundary conditions. In particular, following e.g.\ the analysis in~\cite{NT87, AANS98}, a strip-shaped magnetic billiard can indeed be thought as a model for the quantum Hall effect, in the absence of other interactions. Future work will be devoted to this topic~\cite{quanHall}.

\section*{Acknowledgments}
This work was partially supported by Istituto Nazionale di Fisica Nucleare (INFN) through the project ``QUANTUM'' and the Italian National Group of Mathematical Physics (GNFM-INdAM). 
We thank Giovanni Gramegna for many comments and useful discussions.

\appendix
\section{General theory of quantum boundary conditions}\label{app:generalboundary}
In this appendix we will give a brief sketch of the general theory of quantum BCs for a second-order differential strongly elliptic operator in a smooth, bounded subset $\Omega\subset\mathbb{R}^d$. In full generality, such an operator can be formally written as
\begin{equation}
K=-\sum_{i,j=1}^d\frac{\partial}{\partial x_i}a_{ij}(\x)\frac{\partial}{\partial x_j}+\sum_{i=1}^db_i(\x)\frac{\partial}{\partial x_i}+c(\x)\,,
\end{equation}
where all coefficients are assumed to be smooth up to the boundary $\partial\Omega$. Strong ellipticity means that, for all $\x\in\Omega$, given any $0\neq\xi\in\mathbb{R}^d$ we have
\begin{equation}
	\sum_{i,j=1}^da_{ij}(\x)\xi_i\xi_j>0\,,
\end{equation}
\ie the highest-order coefficients define a positive definite quadratic form. This condition is clearly satisfied by the Laplacian as well as the magnetic Laplacian, both having $a_{ij}(\x)\propto \delta_{ij}$. 

We will also require $K$ to be \emph{formally self-adjoint}, in the following sense: the boundary form $\Lambda_K(\psi,\phi)$, defined in the usual way as
\begin{equation}
	\Lambda_K(\psi,\phi)=\int_\Omega\left[\psi^*K\phi-(K\psi)^*\phi\right]\,\mathrm{d}\x\,,
\end{equation}
vanishes for every test (\ie smooth and compactly supported) functions $\phi, \psi \in \mathcal{D}(\Omega)$.

For any choice of the coefficients, $K$ can be implemented as an operator mapping $\testO$ on itself. Besides, a generalized version of Green's formula holds~\cite{Lions}: letting $\gamma\colon\testO\rightarrow\testdO$ be the restriction operator defined in Eq.~\eqref{eq:gamma}, there exists a boundary operator $\nu_K\colon\testO\rightarrow\testdO$, obtained as a coefficient-dependent linear combination of $\gamma$ and the normal derivative $\nu$, such that
\begin{equation}
\Lambda_K(\psi,\phi)=\int_{\partial\Omega}\left[\gamma \psi^*\,\nu_K\phi-(\nu_K\psi)^*\,\gamma\phi\right]\,\,\mathrm{d} s,
\end{equation}
$\mathrm{d}s$ being the induced measure on $\dO$. 

Now, as in the case of the (magnetic) Laplacian, $K$ is not closed \emph{nor} Hermitian as an operator with domain $\testO$. Its Hermiticity can be restored only by imposing proper BCs, \ie by constructing a bulk-to-boundary operator $\B=T_1\gamma-T_2\nu_K$ such that functions in $\ker\B$ are characterized by a vanishing boundary form. Some explicit examples of BCs restoring the Hermiticity have been found for the magnetic Laplacian. However, Hermiticity is \emph{not} a sufficient condition for an operator to have a well-defined spectrum and hence be associable to a physical observable: \emph{essential self-adjointness}, \ie self-adjointness of its closure, is required. The question is: which BCs render $K$ (essentially) self-adjoint?

A thorough discussion of the complete parametrization of admissible BCs, far from the scope of this paper, requires the introduction of \emph{Sobolev spaces} and the extension of the operators $\gamma$ and $\nu_K$ to two operators, known as Sobolev traces, mapping the maximal domain of $K$ (\ie the space of functions $\psi$ such that $K\psi\in L^2(\Omega)$) to proper negative-order Sobolev spaces. Here, without entering in such details, we will only briefly present the main results in the cases $d=1$ and $d>1$.

\subsection{Case \texorpdfstring{$d=1$}{d=1}}
When $d=1$, \ie $\Omega$ consists of finitely many segment lines (some of them possibly being infinite or semi-infinite), the boundary $\dO$ simply consists of a finite number $r$ of points ($r=2$ in the case of one segment), hence $\gamma$ and $\nu_K$ have values in $\mathbb{C}^r$ and, for every BC, $T_1$ and $T_2$ are $r\times r$ complex matrices.

In this setup, a powerful result can be obtained~\cite{AIM15}: 
\begin{Proposition}
	Admissible self-adjoint realizations of $K$ are \textit{bijectively} parametrized by all $r\times r$ unitary matrices in the following way: $U$ is associated with the boundary condition $\B_U$ characterized by
	\begin{equation}
		T_1=\iu(I+U)\,,\qquad T_2=I-U\,,
		\label{eq-tiunitary}
	\end{equation}
	hence with
	\begin{equation}
		\psi\in\ker \B_U\iff \iu(I+U)\gamma\psi=(I-U)\nu_K\psi\,.
		\label{eq-1dboundary}
	\end{equation}
\end{Proposition}
As a result, a bulk-to-boundary operator $\B$ corresponds to an admissible BC \emph{if and only if} its associated boundary matrices $T_i$ can be written as in Eq.~\eqref{eq-tiunitary}. As one may readily verify, admissible BCs for a segment include Robin conditions, whilst chiral conditions cannot obviously be defined in this setup.

Finally, notice that all BCs parametrized by a unitary matrix $U$ such that $I-U$ is invertible (\ie with $U$ not having $1$ as eigenvalue) can be equivalently parametrized as follows:
\begin{equation}
\nu_K\psi=L\gamma\psi\,,
\end{equation}
where $L=\iu(I+U)(I-U)^{-1}$, the Cayley transform  of $U$, is a Hermitian matrix.

\subsection{Case  \texorpdfstring{$d>1$}{d>1}}
The situation is more involved when $d>1$, with $\partial\Omega$ hence being a nontrivial $(d-1)$-dimensional manifold. In fact, a parametrization formally analogous to the one in Eq.~\eqref{eq-1dboundary} can be obtained, but with a huge complication: instead of $\nu_K$, a ``regularized'' operator $\tilde{\nu}_K$~\cite{Gr68,FGL18}, which heavily depends on the geometry of the system, appears. However, among others, two cases of remarkable physical interest may be proven to define admissible BCs:
\begin{Proposition}
A boundary condition of the form
\begin{equation}
	\nu_K\psi=L\gamma\psi
\end{equation}
is an admissible boundary condition in the following two cases:
\begin{itemize}
	\item $L$ is an \emph{elliptic (pseudo)differential} operator with degree at most $1$~\cite{grubbhassi} (Theorem 8.9);
	\item $L$ is a \emph{bounded} operator~\cite{behrndthassi} (Theorem 6.24 and Corollary 6.25).
\end{itemize}
\end{Proposition}
The first case obviously contains magnetic chiral (and, in particular, Robin) boundary conditions, considered in the main text.

\bibliography{bibfile}

\end{document}